\documentclass[graybox]{svmult}

\usepackage{mathptmx}       
\usepackage{helvet}         
\usepackage{courier}        
\usepackage{makeidx}         
\usepackage{graphicx}        
\usepackage{multicol}        
\usepackage[bottom]{footmisc}
\usepackage{url}
\usepackage{amsfonts,amssymb,amsmath,amsgen}
\usepackage{hyperref,color}
\usepackage{pdfsync}
\usepackage{bbm}

\newtheorem{hyp}{Assumption}

\makeindex             

\begin{document}

\title*{Dispersive estimates for Schr\"odinger operators with point interactions in $\mathbb{R}^3$ }
\author{Felice Iandoli, Raffaele Scandone\thanks{Partially supported by the 2014-2017 MIUR-FIR grant ``\emph{Cond-Math: Condensed Matter and Mathematical Physics}'', code RBFR13WAET and by Gruppo Nazionale per la Fisica Matematica (GNFM-INdAM).}}
\institute{Felice Iandoli \at International School for Advanced Studies – SISSA, via Bonomea 265 34136 Trieste, Italy  \email{fiandoli@sissa.it}
\and Raffaele Scandone \at International School for Advanced Studies – SISSA, via Bonomea 265 34136 Trieste, Italy  \email{rscandone@sissa.it}}

\maketitle

\abstract*{The study of dispersive properties of Schr\"odinger operators with point interactions is a fundamental tool for understanding the behavior of many body quantum systems interacting with very short range potential, whose dynamics can be approximated by non linear Schr\"odinger equations with singular interactions. In this work we proved that, in the case of one point interaction in $\mathbb{R}^3$, the perturbed Laplacian satisfies the same $L^p-L^q$ estimates of the free Laplacian in the smaller regime $q\in[2,3)$. These estimates are implied by a recent result concerning the $L^p$ boundedness of the wave operators for the perturbed Laplacian. Our approach, however, is more direct and relatively simple, and could potentially be useful to prove optimal weighted estimates also in the regime $q\geq 3$. }

\abstract{The study of dispersive properties of Schr\"odinger operators with point interactions is a fundamental tool for understanding the behavior of many body quantum systems interacting with very short range potential, whose dynamics can be approximated by non linear Schr\"odinger equations with singular interactions. In this work we proved that, in the case of one point interaction in $\mathbb{R}^3$, the perturbed Laplacian satisfies the same $L^p-L^q$ estimates of the free Laplacian in the smaller regime $q\in[2,3)$. These estimates are implied by a recent result concerning the $L^p$ boundedness of the wave operators for the perturbed Laplacian. Our approach, however, is more direct and relatively simple, and could potentially be useful to prove optimal weighted estimates also in the regime $q\geq 3$. }

\section{Introduction}
In quantum mechanics, a huge variety of phenomena are described by system of quantum particles interacting with a very short range potentials, supported near away a discrete set of points in $\mathbb{R}^d$. This leads to the study of Hamiltonians which formally are defined as
\begin{equation}\label{heu}
H_{\mu,Y}=``-\Delta+\sum_{y\in Y}\mu_{j}\delta_y"
\end{equation}
where $-\Delta$ is the free Laplacian on $\mathbb{R}^d$, $Y:=\{y_1,y_2,\ldots\}$ is a countable discrete subset of $\mathbb{R}^d$ and $\mu_{y_j}$ are real coupling constants. Thus $H$ describes the motion of a quantum particle interacting with a ``contact potentials", created by point sources of strength $\mu_{y_j}$ centered at $y_j$. 
The first appearance of such Hamiltonians dates back to the celebrated paper of Kronig and Penney \cite{KP}, where they consider the case $d=1$, $Y=\mathbb{Z}$ and  $\mu_{y}$ independent on $y$ as a model of a nonrelativistic electron moving in a fixed crystal lattice. The mathematical rigorous study of $H_{\mu,Y}$ was initiated by Albeverio, Fenstad and H\o egh-Khron \cite{AFH}, and subsequently continued by other authors (see for instance \cite{Z}, \cite{GHM}, \cite{GHM1}, \cite{DG}). In this work we focus on the case of finitely many point interactions on $\mathbb{R}^3$. The rigorous definition of $H_{\mu,Y}$ is based on the theory of self adjoint extension of symmetric operators (see \cite{AGHH} for a complete and detailed discussion): one starts with
\begin{equation}
\tilde{H}_Y:=-\Delta|\mathcal{C}_0^{\infty}(\mathbb{R}^3\backslash\{ Y\}),
\end{equation}
which is a densely defined, non-negative, symmetric operator on $L^2(\mathbb{R}^3)$, with deficiency indices $(N,N)$, and hence it admits a $N^2$-parameter family of self adjoint extension. Among these, we find the important subfamily of the so called local extension, characterized by the following proposition (see \cite{AGHH}, \cite{DMSY}):
\begin{proposition}
Fix $Y:=(y_1,\ldots ,y_N)\subset\mathbb{R}^3$ and $\alpha:=(\alpha_1,\alpha_2,\ldots ,\alpha_N)\in(-\infty,+\infty]^N$. Given $z\in\mathbb{C}$, define
\begin{equation}
G_z(x):=\frac{e^{iz|x|}}{4\pi|x|},\qquad
\tilde{G}_z(x):=\begin{cases}
\frac{e^{iz|x|}}{4\pi|x|}&x\neq 0\\
0&x=0
\end{cases}
\end{equation}
and the $N\times N$  matrix
\begin{equation}
[\Gamma_{\alpha,Y}(z)]_{(j,l)}:=\left[\left(\alpha_j-\frac{iz}{4\pi}\right)\delta_{j,l}-\tilde{G}_z(y_j-y_l)\right]_{(j,l)}
\end{equation}
The meromorphic function $z\mapsto[\Gamma_{\alpha,Y}(z)]^{-1}$ has at most $N$ poles in the upper half space $\mathbb{C}^+$, which are all located along the positive imaginary semi-axis. We denote by $\mathcal{E}$ the set of such poles.
There exists a self adjoint extension $H_{\alpha,Y}$ of $\tilde{H}_{Y}$ with the following properties:
\begin{itemize}
\item Given $z\in\mathbb{C}^+\backslash\{\mathcal{E}\}$, the domain of $H_{\alpha,Y}$ can be written as:
\begin{equation}\label{domain}
\mathcal{D}(H_{\alpha ,Y})=\left\{ \psi:=\phi_z+\sum_{j,l=1}^N(\Gamma_{\alpha,Y}(z)^{-1})_{j,l}\phi_z(j_l)G_z(\cdot-{y_j})\,,\,\phi_{z}\in H^2\right\}.
\end{equation}
The decomposition is unique for a given $z$.
\item With respect to the decomposition \eqref{domain}, the action of $H_{\alpha ,Y}$ is given by
\begin{equation}
(H_{\alpha,Y}-z^2)\psi:=(-\Delta-z^2)\phi.
\end{equation}
\end{itemize}
\end{proposition}
\begin{remark}
The family of self adjoint operator $H_{\alpha,Y}$ realizes in a rigorous way the heuristic definition given by expression \eqref{heu}. It is worth noticing the different roles played by parameters: while $\mu_j$ measures the strength of the point  interactions at $y_j$,  $\alpha_j$ is related to the scattering length.  Indeed, a generic function $\psi\in\mathcal{D}(H_{\alpha ,Y})$, $\alpha\neq 0$, satisfies the so called \emph{Bethe-Peierls contact condition}
\begin{equation}\label{beth}
\psi(x)\mathop{\sim}_{x\rightarrow y_i}\frac{1}{|x-y_j|}+4\pi\alpha_j,\quad j=1,\ldots ,N
\end{equation}
which is typical  for the low-energy behavior of an eigenstate of the Schr\"{o}dinger equation for a quantum particle subject to a very short range potentials, centered at $y_j$ and with $s$-wave scattering length $-(4\pi\alpha)^{-1}$ ( see the works of  Bethe and Peierls \cite{BP1}, \cite{BP2}).
When $\alpha_j=+\infty$, no actual interactions take place at $y_j$  (the s-wave has infinity scattering length); in particular when $\alpha=+\infty$ we recover the Friedrichs extension of $\tilde{H}_{Y}$, namely the free Laplacian on $L^2(\mathbb{R}^3)$.   
\end{remark}
The spectral properties of $H_{\alpha,Y}$ are well known and completely characterized; we encode them in the following proposition (see \cite{AGHH}, \cite{DMSY}):
\begin{proposition}~
\begin{enumerate}
\item The spectrum $\sigma(H_{\alpha,Y})$ of $H_{\alpha,Y}$ consists of at most $N$ negative eigenvalues and the absolutely continuous part $\sigma_{ac}(H_{\alpha,Y})=[0,+\infty)$. Moreover, there exists a one to one correspondence between the poles $i\lambda\in\mathcal{E}$ and the negative eigenvalues $-\lambda^2$ of $H_{\alpha,Y}$, counted with multiplicity.
\item The resolvent of $H_{\alpha,Y}$ is a rank $N$ perturbation of the free resolvent, and it is given by:
\begin{equation}
(H_{\alpha,Y} -z^2)^{-1} -(H_0-z^2)^{-1} \;=\; 
\sum_{j,k=1}^N  (\Gamma_{\alpha,Y}(z)^{-1})_{jk} \,G_{z}^{y_j}\otimes   
\overline{G_{z}^{y_k}}.
\end{equation}
\end{enumerate}
\end{proposition}
We conclude this introduction by observing that $H_{\alpha,Y}$ can be also realized as limit of scaled short range Schr\"odinger operator. Indeed we have the following Proposition (see \cite{AGHH}):
\begin{proposition}\label{approris}

Fix $\alpha\in (-\infty, +\infty]^N$ and $Y=\{y_1,\ldots ,y_N\}\subset\mathbb{R}^3$. There exists real valued potential $V_1,\ldots V_N$ of finite Rollnik norm, and real analytic functions $\lambda_j:\mathbb{R}\rightarrow\mathbb{R}$, with $\lambda_j(0)=1$, such that the family of Schr\"odinger operators
\begin{equation}
H_{\epsilon}:=-\Delta+\sum_{j=1}^N\frac{\lambda_j(\epsilon)}{\epsilon^{2}}V\left(\frac{x-y_j}{\epsilon}\right)
\end{equation}
converges in strong resolvent sense to $H_{\alpha,Y}$ as $\epsilon$ goes to zero. Moreover:
\begin{equation}\label{neceriso}
 \alpha_j\neq +\infty \mbox{ for some } j \iff -\Delta+V_{j}\mbox{ has a zero energy resonance}.
\end{equation}
\end{proposition}
\begin{remark}
Proposition \ref{approris} makes more convincing the idea of considering the Hamiltonian $H_{\alpha,Y}$ as an approximation of more realistic phenomena, governed by very short range interactions.
\end{remark}

\section{Dispersive properties of $H_{\alpha,Y}$}
Since $H_{\alpha}$ is a self adjoint operator, it generates a unitary group of operators $e^{itH_{\alpha,Y}}$; in particular the $L^2$ norm is preserved by the evolution:
\begin{equation}\label{triv_bound}
\Vert e^{itH_{\alpha,Y}}f\Vert_{L^2(\mathbb{R}^3)}=\Vert f\Vert_{L^2(\mathbb{R}^3)}.
\end{equation}
It is natural to investigate the dispersive properties of $e^{itH_{\alpha,Y}}$. The first work in this direction is by D'ancona, Pierfelice and Teta \cite{DPT}, who proved weighted $L^1-L^{\infty}$ estimates
\begin{equation}\label{wei}
\Vert w^{-1}e^{itH_{\alpha,Y}}P_{ac}f\Vert_{L^{\infty}(\mathbb{R}^3)}\lesssim t^{-\frac32}\Vert wf\Vert_{L^1(\mathbb{R}^3)}
\end{equation}
where $P_{ac}$ is the projection onto the absolutely continuous sprectrum of $H_{\alpha,Y}$ and
\begin{equation}
w(x)=\sum_{j=1}^N\left(1+\frac{1}{|x-y_j|}\right),
\end{equation}
under the following assumption:
\begin{hyp}\label{assu}
The matrix $\Gamma_{\alpha,Y}(z)$ is invertible for $z\geq 0$, with locally bounded inverse.
\end{hyp}
It is worth noticing that the presence of a weight in \eqref{wei} is  unavoidable, because of the singularities appearing in the domain of $H_{\alpha,Y}$. In the case of one single point interaction, assumption \ref{assu} is always satisfied except for $\alpha=0$,  in which case the perturbed Hamiltonian has a zero energy resonance. Nevertheless, exploiting the explicit formula for the propagator $e^{itH}$ available in the case $N=1$ ( see \cite{ST} \cite{ABD}), also the case $\alpha=0$  was settle down in \cite{DPT}, by showing weighted dispersive inequality with a slower decay in $t$, a typical phenomenon for Schr\"odinger operators with zero energy resonances:
\begin{equation}\label{ri_so}
\Vert w^{-1}e^{itH_{0,y}}f\Vert_{L^{\infty}(\mathbb{R}^3)}\lesssim t^{-\frac12}\Vert wf\Vert_{L^1(\mathbb{R}^3)}.
\end{equation}
Observe now that, interpolating \eqref{wei} and \eqref{ri_so} with the trivial bound \eqref{triv_bound}, we get weighted dispersive inequlities in the full range $q\in[2,+\infty]$:
\begin{proposition}~
\begin{enumerate}
\item
Under assumption \ref{assu}, the following estimates holds:
\begin{equation}\label{full_fiacco}
\Vert w^{-\left(1-\frac2q\right)}e^{itH_{\alpha,Y}}P_{ac}f\Vert_{L^{q}(\mathbb{R}^3)}\lesssim t^{-\frac32(\frac1p-\frac1q)}\Vert w^{\frac2p-1}f\Vert_{L^p(\mathbb{R}^3)}
\end{equation}
where $q\in[2,+\infty]$ and $p$ is the dual exponent of $q$.
\item In the case $N=1$, $\alpha=0$ we have
\begin{equation}\label{full_ris_fiacco}
\Vert w^{-\left(1-\frac2q\right)}e^{itH_{0,y}}f\Vert_{L^{q}(\mathbb{R}^3)}\lesssim t^{-\frac12\left(\frac1p-\frac1q\right)}\Vert w^{\frac2p-1}f\Vert_{L^p(\mathbb{R}^3)}
\end{equation}
where $q\in[2,+\infty]$ and $p$ is the dual exponent of $q$.
\end{enumerate}
\end{proposition}
However, since the singularities $G_{i}(x-y_j)$ belong to $L^q(\mathbb{R}^3)$ for $q\in[2,3)$, one may hope, at least in principle, to prove an unweighted version of \eqref{full_fiacco} and \eqref{full_ris_fiacco}. This is true indeed, and it is a consequence of a recent result \cite{DMSY}:
\begin{theorem}\label{wavebound}
For any $Y$ and $\alpha$, the wave operators 
\begin{equation}
 W^{\pm}_{\alpha,Y}=s-\lim_{t\rightarrow +\infty}e^{itH_{\alpha ,Y}}e^{it\Delta}
\end{equation}
for the pair $(H_{\alpha,Y},-\Delta)$ exists and are complete on $L^2(\mathbb{R}^3)$, and they are bounded on
 $L^q(\mathbb{R}^3)$ for $1<q<3$.
\end{theorem}
\begin{remark}
The restriction $1<q<3$ already emerges at level of approximating Schr\"odinger operators. Indeed, if $H=-\Delta+V$ has a zero energy resonance (which by Proposition \ref{approris} is a necessary condition for $H_{\epsilon}$ to converges to $H$), then the wave operators 
\begin{equation}
W^{\pm}_{V}:=s-\lim_{t\rightarrow +\infty}e^{it(-\Delta+V)}e^{it\Delta}
\end{equation}
are bounded on $L^q$ if and only if $1<q<3$  (see Yajima \cite{Y})
\end{remark}
Owing to Theorem \ref{wavebound} and the intertwining property of wave operators, viz.
\begin{equation}\label{interwining}
 f(H_{{\alpha},Y})P_{ac}H_{{\alpha},Y}= W^{\pm}_{\alpha,Y}f(-\Delta) (W^{\pm}_{\alpha,Y})^{\ast}
\end{equation}
for any Borel function $f$ on $\mathbb{R}^3$, one can lift the classical dispersive estimates for the free Laplacian into analogous estimates for $H_{\alpha,Y}$, albeit for the restriction on the exponent $q$. Thus we find:
\begin{proposition}\label{result}
For any $\alpha$ and $Y$, we have the estimate
\begin{equation}\label{unwei}
 \Vert e^{itH_{\alpha,Y}}P_{ac}f\Vert_{L^{q}(\mathbb{R}^3)}\lesssim t^{-\frac32(\frac1p-\frac1q)}\Vert f\Vert_{L^p(\mathbb{R}^3)}\qquad\mbox{for }q\in[2,3).
\end{equation}
\end{proposition}
Interpolating \eqref{unwei} respectively with \eqref{wei} and \eqref{ri_so}, we deduce also the following:
\begin{corollary}\label{almootti}
 \begin{enumerate}
  \item Under assumption \ref{assu}, we have
\begin{equation}\label{full}
 \Vert w^{-\left(1-\frac{3-\epsilon}{q}\right)}e^{itH_{\alpha,Y}}P_{ac}f\Vert_{L^{q}(\mathbb{R}^3)}\lesssim t^{-\frac32(\frac1p-\frac1q)}\Vert w^{\left(1-\frac{3-\epsilon}{q}\right)}f\Vert_{L^p(\mathbb{R}^3)}
\end{equation}
in the regime $q\in[3,+\infty]$.
\item When $N=1$ and $\alpha=0$, we have
\begin{equation}\label{full_ris}
\Vert w^{-\left(1-\frac{3-\varepsilon}{q}\right)}e^{itH_{0,y}}f\Vert_{L^{q}(\mathbb{R}^3)}\lesssim t^{-\frac12+\frac{\epsilon}{q}}\Vert w^{\left(1-\frac{3-\epsilon}{q}\right)}f\Vert_{L^p(\mathbb{R}^3)}
\end{equation}
in the regime $q\in[3,+\infty]$.
\end{enumerate}
\end{corollary}
We can see that the results in \cite{DMSY} improves the one in \cite{DPT} in various ways:
\begin{enumerate}
 \item In the regime $q\in[2,3)$, with an arbitrary number of centers, both the weights and the hypothesis \ref{assu} on $\Gamma$ are removed.
\item In the regime $q\in[3,+\infty]$, with an arbitrary number of centers and under the hypothesis \ref{assu}, the weights are strengthened to be almost optimal (indeed we can not remove $\epsilon$ in estimate \eqref{full}).
\item In the case $N=1$, $\alpha=0$ and in the regime $q\in[2,3)$, weights are removed and the time decay is strengthened.
\item In the case $N=1$, $\alpha=0$ and in the regime $q\in[3,+\infty]$, both the weights and the time decay are strengthened, and again the weights are almost optimal.
\end{enumerate}
In this work we want to provide a new and simpler proof of Proposition \ref{result} in the particular case $N=1$, without using any information about the wave operators. 
\section{Proof of Proposition \ref{result}, case $N=1$}
\label{sec:3}
The operators $H_{\alpha,y_1}$ and $H_{\alpha,y_2}$ are conjugated by translations, hence we can assume $y=0$ and we will simply write $H_{\alpha}$ instead of $H_{\alpha,0}$. We recall an useful factorization for the operator $H_{\alpha}$ (see \cite{AGHH}). Introducing spherical coordinates on $\mathbb{R}^3$, we can decompose $L^2(\mathbb{R}^3)$ with respect to angular momenta:
\begin{equation}
 L^2(\mathbb{R}^3)=L^2(\mathbb{R}^+,r^2dr)\otimes L^2(S^2)
\end{equation}
where $S^2$ is the unit sphere in $\mathbb{R}^3$. Moreover, using the unitary transformation
\begin{equation}
 U:L^2((0,+\infty),r^2dr)\rightarrow L^2(\mathbb{R}^+,rdr),\quad (Uf)(r)=rf(r)
\end{equation}
and decomposing $L^2(S^2)$ into spherical harmonics 
\begin{equation}
\left\{Y_{l,m}l\in\mathbb{N},m=0,\pm 1,\ldots, \pm l\right\},
\end{equation}
we obtain
\begin{equation}\label{finde}
L^2(\mathbb{R}^3)=\bigoplus_{l=0}^{+\infty}U^{-1}L^2(\mathbb{R}^+,rdr)\otimes \langle Y_{l,-l},\ldots ,Y_{l,l}\rangle.
\end{equation}
With respect to this decomposition, the symmetric operator $\tilde{H}:=\tilde{H}_{\{0\}}$ writes as
\begin{equation}
\tilde{H}=\bigoplus_{l=0}^{+\infty}U^{-1}h_lU\otimes 1
\end{equation}
where $h_l$, $l\geq 0$ are symmetric operators on $L^2(\mathbb{R}^+)$, with common domain $\mathcal{C}_0^{+\infty}(\mathbb{R}^+)$ and actions given by
\begin{equation}
h_l=-\frac{d^2}{dr^2}+\frac{l(l+1)}{r^2},\quad r>0.
\end{equation}
 It is well known \cite{RS} that $h_l$ are essentially self adjoint for $l\geq 1$, while $h_0$ admits a one parameter family of selfadjoint extension $\dot{h}_{0,\alpha}$ such that
\begin{equation}\label{de:se}
H_{\alpha}=\dot{h}_{0,\alpha}\oplus\bigoplus_{l=1}^{+\infty}U^{-1}\dot{h}_lU\otimes 1
\end{equation}
where $\dot{h}_l$ is the unique self adjoint extension of $h_l$, for $l\geq 1$. Identity \eqref{de:se} tells us that $H_{\alpha}$ completely diagonalizes with respect to decomposition \eqref{finde}, and that it coincides with $-\Delta$ after projecting out the subspace of radial functions. Hence it immediately follows
\begin{lemma}\label{orto}
Suppose $f\in L^2(\mathbb{R}^3)$ is orthogonal to the subspace of radial functions. Then
\begin{equation}
e^{itH_{\alpha}}f=e^{-it\Delta}f
\end{equation}
\end{lemma}
Lemma \ref{orto} has an important Corollary, which considerably simplifyes our proof:
\begin{corollary}\label{radial}
In the proof of Proposition \ref{result} (in the special case $N=1$) we can suppose $f$ to be radial.
\end{corollary}
\begin{proof}
Suppose \eqref{unwei} to be true for radial functions. Given a generic $f\in L^2(\mathbb{R}^3)$, we can decompose it as $f_1+f_2$, where
\begin{equation}
f_1:=\frac{4\pi}{|y|^2}\int_{S_y}f(r,\omega)d\mathcal{H}^2(\omega)
\end{equation}
is the orthogonal projection onto $L^2_{rad}(\mathbb{R}^3)$. By Lemma \ref{orto}, we get
\begin{equation}
 e^{itH_{\alpha}}f=e^{itH_{\alpha}}f_1+e^{-it\Delta}f_2.
\end{equation}
By hypothesis and using the dispersive estimates for the free Laplacian, we deduce
\begin{equation}
\Vert e^{itH_{\alpha}}f \Vert_{L^q}\leq t^{\frac32\left(\frac 1p-\frac 1q\right)}(\Vert f_1\Vert_{L^p}+\Vert f_2\Vert_{L^p}).
\end{equation}
Now, using H\"older inequality, we get
\begin{align*}
\Vert f_1\Vert_{L^p}^p&\leq\int_0^{+\infty}r^{2-2p}\left(\int_{S_r}|f(r,\omega)|d\mathcal{H}^2(\omega)\right)^pdr\\
&\leq \int_0^{+\infty}r^{2-2p+\frac{2p}{q}}\int_{S_r}|f(r,\omega)|^pd\mathcal{H}^2(\omega)dr=\\
&\int_0^{+\infty}\int_{S_r}||f(r,\omega)|d\mathcal{H}^2(\omega)|^p=\Vert f\Vert_{L^p}^p
\end{align*}
and consequently
\begin{equation}
\Vert f_2\Vert_{L^p}\leq \Vert f\Vert_{L^p}+\Vert f_1\Vert_{L^p}\lesssim \Vert f\Vert_{L^p}
\end{equation}
which concludes the proof.
\end{proof}
Now we  are in turn to prove our main result. As mentioned before, in the case $N=1$, the propagator associated to $H_{\alpha}$ is explicitly known. In particular, Scarlatti and Teta \cite{ST} have proved the following characterization:
\begin{equation}\label{propagatore}
e^{itH_{\alpha}}f=\begin{cases}
e^{-it\Delta}f+ \lim_{R\rightarrow\infty}M_R f & \mbox{ if } \alpha=0  \\
e^{itH_{0}}f+ \lim_{R\rightarrow\infty}M_{\alpha,R} f & \mbox{ if } \alpha>0 \\
e^{itH_{0}}f+ \lim_{R\rightarrow\infty}\widetilde{M}_{\alpha,R} f & \mbox{ if } \alpha<0  \\
\end{cases}
\end{equation}
where the limit is taken in the $L^2$ sense and 
\begin{equation}
M_Rf(x):=(4\pi it)^{-{1/2}}\frac{1}{4\pi|x|}\int_{B_R}\frac{\widetilde{f}(|y|)}{|y|}e^{-i\frac{(|x|+|y|)^2}{4t}}dy,
\end{equation}
\begin{equation}
M_{\alpha,R}f(x):=-(4\pi it)^{-1/2}\frac{\alpha}{|x|}\int_{\mathbb{R}^3}\frac{f(y)}{|y|}\int_{0}^{+\infty}e^{-4\pi\alpha s}e^{-i\frac{(|x|+|y|+s)^2}{4t}}dsdy,
\end{equation}
\begin{equation}\label{alphaneg}
\begin{aligned}
\widetilde{M}_{\alpha,R}&f(x):=\Bigg(-\psi_{\alpha}(x)\int_{B_R}\psi_{\alpha}(y)f(y)e^{it(4\pi\alpha)^2}dy\\
&-\frac{\alpha}{|x|}(it\pi)^{-1/2}\int_{B_R}\frac{f(y)}{|y|}\int_0^{+\infty}e^{4\pi\alpha s}\exp\left(-\frac{(u-|x|-|y|)^2}{4it}\right)dsdy\Bigg),
\end{aligned}\end{equation}
and $\psi_{\alpha}(x)=\sqrt{-2\alpha}\frac{e^{4\pi\alpha |x|}}{|x|}$ is the normalized eigenfunction associated to the negative eigenvalue $-(4\pi\alpha)^2$ for $\alpha<0$. 
We are going to show that the following estimates hold uniformly in $R>0$:
\begin{align}
\label{riso}&\|M_{R}f\|_{L^q}\lesssim t^{-\frac{3}{2}(\frac1p-\frac1q)}\|f\|_{L^p},\\
\label{nonriso}&\|M_{\alpha,R}f\|_{L^q}\lesssim t^{-\frac{3}{2}(\frac1p-\frac1q)}\|f\|_{L^p},\\
\label{nonrisoneg} &\| \widetilde{M}_{\alpha,R}P_{ac}f\|_{L^q}\lesssim t^{-\frac{3}{2}(\frac1p-\frac1q)}\|f\|_{L^p}.
\end{align}
The latter inequalities are clearly sufficient to prove Proposition \ref{result} in the special case $N=1$.
Let us start by proving inequality \eqref{riso}. Thanks to Corollary \ref{radial} we can suppose $f(y)=\widetilde{f}(|y|)$ for some $\widetilde{f}:\mathbb{R}\rightarrow\mathbb{R}$. Using spherical coordinates in both variables $x$ and $y$ we get
\begin{equation}
\|M_{R}f\|_{L^q}\lesssim t^{-1/2}\left[\int_0^{+\infty}r^{2-q}\left|\int_0^R \exp\left(-i\frac{\rho r}{2t}-i\frac{\rho^2}{4t}\right)\rho \widetilde{f}(\rho)d\rho\right|^qdr\right]^{1/q}.
\end{equation}
Setting
\begin{equation}
h(\rho):=\left\{\begin{array}{cl}
e^{-i\rho^2/4t}\rho\widetilde{f}(\rho)\,\,\,\, & 0\leq\rho\leq R \\
0 & \rho\in\mathbb{R}\setminus[0,R] \\
\end{array}\right.
\end{equation}
the latter expression becomes
\begin{equation}
t^{-1/2}\left[\int_0^{+\infty}r^{2-q}\left|\widehat{h}\left(\frac{r}{2t}\right)\right|^qdr\right]^{1/q},
\end{equation}
which is equal to
\begin{equation}
t^{-\frac32(\frac1p-\frac1q)}\left[\int_0^{+\infty}r^{2-q}|\widehat{h}(r)|^qdr\right]^{1/q}.
\end{equation}
At this point we are ready to use a classical weighted Fourier transform norm inequality, also known in literature as Pitt's inequality. We state here the original theorem proved by Pitt in 1937 \cite{P}:
\begin{theorem}\label{Pittineq}[Pitt's theorem]
Let $1<\gamma\leq \eta<\infty$, choose $0<b<\frac{1}{\gamma'}$ with $\frac{1}{\gamma}+\frac{1}{\gamma'}=1$, set $\beta=1-\frac{1}{\gamma}-\frac{1}{\eta}-b<0$ and define $v(x)=|x|^{b\gamma}$ for all $x\in\mathbb{R}$. There is a constant $C>0$ such that
\begin{equation}
\left(\int_{\mathbb{R}}|\widehat{f}(\xi)|^{\eta}|\xi|^{\beta \eta}d\xi\right)^{1/{\eta}}\leq C\left(\int_{\mathbb{R}}|f(x)|^{\gamma}|x|^{b\gamma}dx\right)^{1/\gamma},
\end{equation}
for all $f\in L^{\gamma}_v(\mathbb{R})$.
\end{theorem}
Since $q<3$ we may use this Theorem in the case $\eta=q$, $\gamma=p$, $\beta=\frac{2-q}{q}$ and $b=\frac{2-p}{p}$ obtaining 
\begin{equation}
t^{-\frac32(\frac1p-\frac1q)}\left[\int_0^{+\infty}r^{2-q}|\widehat{h}(r)|^qdr\right]^{1/q}\lesssim t^{-\frac32(\frac1p-\frac1q)}\left[\int_0^{+\infty}|h(r)|^pr^{2-p}dr\right]^{1/p},
\end{equation}
which essentially is the desired estimate, indeed
\begin{equation}
\left[\int_0^{+\infty}|h(r)|^pr^{2-p}dr\right]^{1/p}=\|f\|_{L^p}.
\end{equation}
This concludes the proof of \eqref{riso}, which, together with the standard dispersive estimates for the free Laplacian, implies the dispersive estimates for the semigroup $\{e^{itH_0}\}_{t>0}$.

Let us turn in to proving \eqref{nonriso}. Since $q<3$ the function $1/|y|$ belongs to $L^q(B_R)$, hence we can exchange the order of integration and  use Minkowski inequality
\begin{equation}
\begin{aligned}
&\|M_{R}f\|_{L^q}\lesssim \\
&\lesssim t^{-1/2}\int_0^{+\infty}\left\|\int_{B_R}\frac{1}{|x|}\exp\left({-4\pi\alpha s-i\frac{(|x|+|y|+s)^2}{4t}}\right)\frac{\widetilde{f}(y)}{|y|}dy\right\|_{L^q} ds\\
&=t^{-1/2}\int_0^{+\infty}e^{-4\pi\alpha s}\left\|\int_{B_R}\frac{1}{|x|}\exp\left(-i\frac{|y|^2}{4t}-i\frac{|x||y|}{2t}-i\frac{s|y|}{2t}\right)\frac{f(y)}{|y|}dy\right\|_{L^q}ds,
\end{aligned}\end{equation}
which, as before, in spherical coordinates is bounded, up to constants, by
\begin{equation}\label{cane}
t^{-1/2}\int_0^{+\infty}e^{-4\pi\alpha s}\left(\int_0^{+\infty}r^{2-q}\left|\int_0^Re^{-i\frac{r\rho}{2t}} h_s(\rho)(\rho)d\rho\right|^{q}dr\right)ds,
\end{equation}
where 
\begin{equation}
h_s(\rho):=\left\{\begin{array}{cl}
\exp\left(-i\frac{\rho^2}{4t}-i\frac{s\rho}{2t}\right)\rho \widetilde{f}(\rho)\,\,\,\, & 0\leq\rho\leq R \\
0 & \rho\in\mathbb{R}\setminus[0,R] \\
\end{array}\right..
\end{equation}
The  quantity \eqref{cane} is nothing but 
\begin{equation}
t^{-1/2}\int_0^{+\infty}e^{-4\pi\alpha s}\left(\int_0^{+\infty}r^{2-q}\left|\widehat{h_s}\left(\frac{r}{2t}\right)\right|^qdr\right)^{1/q}ds,
\end{equation}
which, arguing as before,  is bounded by $t^{-\frac32 (\frac1p-\frac1q)}\|f\|_{L^p}$. This concludes the proof of \eqref{nonriso}.

 The proof  of \eqref{nonrisoneg} is very similar, indeed after projecting $f$ onto the the absolutely continuous spectrum of $H_{\alpha}$, the first summand in the right hand side of \eqref{alphaneg} disappears  and hence the remaining part can be treated exactly in the same way as done in the proof of \eqref{nonriso}.
\section{Conclusions}
\label{sec:4}
The proof given in section \ref{sec:3} is quite direct and does not use any deep results from scattering theory for the perturbed Hamiltonian $H_{\alpha,Y}$. Nevertheless, it is worth noticing that the proof of Pitt's inequality, the main tool of our argument, requires some technical results from harmonic analysis such as Muckenhaupt estimates (\cite{G} \cite{M}), which play an essential role also in the proof of the $L^p$ boundedness of the wave operators $W^{\pm}$ given in \cite{DMSY}. The main advantage of our approach is that, owing to more general  weighted Fourier inequalities (see for instance \cite{BH}, \cite{H}) rather than Pitt's inequality (in which the weights are forced to be pure powers), it can potentially be adapted to prove optimal $L^p-L^q$ estimates also in the regime $q\geq 3$. In particular, we conjecture the following result:
\begin{conjecture}\label{conuno}
Fix $q\in [3,+\infty]$, and let $w_{q}(x)$ a weight such that $w(x)\equiv 1$ outside a ball centered at the origin and $w_{q}^{-1}G_{i}\in L^q(\mathbb{R}^3)$. Then for every $\alpha\neq 0$ and $y\in\mathbb{R}^3$, the following estimates holds:
\begin{equation}
\Vert w_{q}(\cdot-y)^{-1}e^{itH_{\alpha,y}}P_{ac}f\Vert_{L^q(\mathbb{R}^3)}\lesssim t^{-\frac32\left(\frac1p-\frac1q\right)}\Vert w_{q}(\cdot-y)f\Vert_{L^p(\mathbb{R}^3)}.
\end{equation}
When $\alpha=0$, a similar estimates holds but with a slower time decay:
\begin{equation}
\Vert w_{q}(\cdot-y)^{-1}e^{itH_{\alpha,y}}f\Vert_{L^q(\mathbb{R}^3)}\lesssim t^{-\frac 12}\Vert w_{q}(\cdot-y)f\Vert_{L^p(\mathbb{R}^3)}.
\end{equation}
\end{conjecture}
\begin{remark}
Conjecture \ref{conuno} is motivated by the natural principle for which removing the local singularity is enough to get dispersive estimates, and it would improve the result in  Corollary \ref{almootti}. For example when $q=3$ we expect that a logarithmic weight would suffice, while in estimates \eqref{full} and \eqref{full_ris} there appear polynomial weights.
\end{remark}
An alternative conjecture can be expressed in term of weighted Lorentz space, in which context there are other generalizations of Pitt's inequality (see for instance \cite{S}):
\begin{conjecture}
Given $q\in [3,+\infty]$, define the weight
$$w_q:=1+|x|^{\frac{3}{q}-1}$$
Then for every $\alpha\neq 0$ and $y\in\mathbb{R}^3$, the following estimates holds:
\begin{equation}\label{lowf}
\left\Vert w_q(\cdot-y)^{-1}e^{itH_{\alpha,y}}P_{ac}f\right\Vert_{L^{q,\infty}(\mathbb{R}^3)}\lesssim t^{-\frac32\left(\frac1p-\frac1q\right)}\left\Vert w_q(\cdot -y)f\right\Vert_{L^{p,1}(\mathbb{R}^3)}.
\end{equation}
When $\alpha=0$, a similar estimates holds but with a slower time decay:
\begin{equation}\label{lowfr}
\left\Vert w_q(\cdot -y)^{-1}e^{itH_{\alpha,y}}f\right\Vert_{L^{q,\infty}(\mathbb{R}^3)}\lesssim t^{-\frac12}\left\Vert w_q(\cdot -y)f\right\Vert_{L^{p,1}(\mathbb{R}^3)}.
\end{equation}
\end{conjecture}
\begin{remark}
The function $w_q^{-1}G_{i}$ belongs to $L^{q,\infty}(\mathbb{R}^3)$, hence the plausibility of the conjecture. Observe moreover that it would be enough to prove \eqref{lowf} and \eqref{lowfr} when $q=3$, the general case following by interpolation with $q=\infty$, in which case we recover the weighted $L^1-L^{\infty}$ estimates \eqref{wei} and \eqref{ri_so} proved in \cite{DPT}.
\end{remark}

%
%
%

\end{document}